\documentclass[copyright,creativecommons]{eptcs}
\providecommand{\event}{GaSCOM 2026} 

\usepackage{amsmath,amssymb,amsthm,mathtools,enumitem}

\newcommand{\Des}{\mathrm{Des}}

\usepackage{diagbox}
\usepackage{makecell}

\usepackage{amssymb,latexsym,cite,longtable,url,array,booktabs}
\usepackage[linguistics]{forest}
\usepackage{multicol}
\usepackage{arydshln,pstricks,pst-node,ytableau}
\usepackage{multicol,float,vcell}
\usepackage{lipsum}
\usepackage{comment}
\usepackage{tabularx,booktabs}
\usepackage{tikz}
\usetikzlibrary{positioning}
\usepackage{hyperref}

\newcolumntype{Z}[0]{>{\centering\arraybackslash}X}
\newcolumntype{n}[0]{>{\hsize=.8\hsize}Z}
\newcolumntype{s}[0]{>{\hsize=.5\hsize}Z}

\makeatletter
\pgfkeys{/ytableau/options,
 onlyboxsize/.value required,
 onlyboxsize/.code = {%
  \compare@YT{#1}{normal}%
  \ifeq@YT
   \xdef\macro@boxdim@YT{\expandonce@YT\boxdim@normal@YT}%
  \else
   \xdef\macro@boxdim@YT{#1}%
  \fi
 }
}
\makeatletter

\newenvironment{smallytableau}{%
  \ytableausetup{smalltableaux,onlyboxsize=1em}%
  \begin{ytableau}%
}{%
  \end{ytableau}%
  \ytableausetup{nosmalltableaux,boxsize=normal}%
}
\usepackage{makecell}

\usepackage{tikz}
\usepackage{xparse}

\makeatletter
\newcommand*{\rom}[1]{\expandafter\@slowromancap\romannumeral #1@}
\makeatother

\newcount\tableauRow
\newcount\tableauCol

\newenvironment{Tableau}[1]{%
  \tikzpicture[scale=0.7,draw/.append style={thick,black},
                      baseline=(current bounding box.center)]
    \tableauRow=-1.5
    \foreach \Row in {#1} {
       \tableauCol=0.6
       \foreach\k in \Row {
         \draw[thin](\the\tableauCol,\the\tableauRow)rectangle++(1,1);
         \draw[thin](\the\tableauCol,\the\tableauRow)+(0.5,0.5)node{$\k$};
         \global\advance\tableauCol by 1
       }
       \global\advance\tableauRow by -1
    }
}{\endtikzpicture}

\makeatletter
\renewcommand{\paragraph}{\@startsection{paragraph}{4}{0pt}%
   {-3.25ex plus -1ex minus -0.2ex}%
   {1.5ex plus 0.2ex}%
   {\normalfont\normalsize\bfseries}}
\makeatother

\stepcounter{secnumdepth}
\stepcounter{tocdepth}

\usepackage{microtype}

\usepackage{nicematrix}
\NiceMatrixOptions
{
 custom-line =
 {
letter = ; ,
command = \hdashedline ,
ccommand = \cdashedline ,
tikz = { dash pattern = on 4 pt off 2pt } ,
total-width = \pgflinewidth
 }
}

\newcommand{\BBB}{{\mathcal{B}}}

\newtheorem{theorem}{Theorem}[section]
\newtheorem{definition}[theorem]{Definition}
\newtheorem{thm}[theorem]{Theorem}

\newtheorem{example}[theorem]{Example}
\newtheorem{cor}[theorem]{Corollary}

\newtheorem{remark}[theorem]{Remark}

\providecommand{\urlalt}[2]{\href{#1}{#2}}
\providecommand{\doi}[1]{doi:\urlalt{http://dx.doi.org/#1}{#1}}

\newcommand{\Q}{{\mathcal Q}}
\newcommand{\F}{\mathcal{F}}

\title{Touchard-Riordan Polynomials and Schur-positivity of Set Partitions}
\author{Eli Bagno and David Garber}

\author{
Eli Bagno
\institute{Jerusalem College of Technology, 21 HaVaad HaLeumi St., Jerusalem, Israel,\\ and  Michlalah College Jerusalem, 36 Barukh Duvdevani St., Jerusalem, Israel}
\email{bagnoe@g.jct.ac.il}
\and
David Garber
\institute{Holon Institute of Technology, 52 Golomb St., P.O.Box 305, 5810201 Holon, Israel}
\email{garber@hit.ac.il}
}

\def\titlerunning{Touchard-Riordan Polynomials and Schur-positivity of Set Partitions}
\def\authorrunning{Eli Bagno and David Garber}

\begin{document}

\maketitle

\begin{abstract}
A symmetric function is called {\it Schur-positive} if it admits an expansion in the Schur basis with nonnegative coefficients. In this paper, we study the Schur-positivity of symmetric functions naturally associated with set partitions, with respect to a descent set function that considers $i$ as  descent, if $i$ and $i+1$ share a block in the partition.
The Schur expansion involves hook-shaped Young diagrams, and the corresponding coefficients are given by Touchard-Riordan polynomials, which enumerate matchings by their number of crossings.
\end{abstract}

\section{Introduction}

Schur-positivity is a central notion in algebraic combinatorics, referring to symmetric functions that can be expressed as a nonnegative linear combination of Schur functions. Since Schur functions form a distinguished basis for the ring of symmetric functions and correspond to irreducible representations of the symmetric group, Schur-positivity often reflects deep algebraic and representation-theoretic structure.
Determining whether a given symmetric function is
Schur-positive is a major problem in contemporary algebraic
combinatorics~\cite{Stanley_problems}.

Given $D \subseteq \{1,\dots,n-1\}$, define $\mathcal{F}_D$ to be the fundamental quasi-symmetric function indexed by $D$ (see Definition \ref{def fund} below).
Given a combinatorial object $\sigma$, one is interested in studying the generating function $\Q_n=\sum \mathcal{F}_{{\rm Des}(\sigma)}$, where ${\rm Des}(\sigma)$ is a descent set function from the set of the objects $\sigma$ is taken from, to the set $\mathbb{N}$.
We look for functions $\Q_n$ which are symmetric or Schur-positive for each $n$.
When the objects under discussion are permutations, this subject has been developed extensively.
Unlike the permutation case, there is still no comparably general framework explaining Schur-positivity for broad families of set partitions, except for some initial results, see \cite{BGMS}.

Set partitions can be encoded by restricted growth words, and this encoding has been used extensively in the literature on statistics of set partitions and their generating functions, see e.g. \cite{CaRe} and \cite[Sec. 1.7]{EgGa}.  This makes the study of Schur-positivity for set partitions a natural extension of the classical permutation setting: one seeks structural conditions on families of set partitions, or on their restricted-growth words, that force the associated quasi-symmetric generating functions to be symmetric and Schur-positive.

\medskip

Adin and Roichman \cite{AR} dealt with the descent-type parameter on matchings (i.e. set partitions having blocks of size at most $2$) which they called {\it geometric descent}: an index $i \in [n-1]$ is a {\it geometric descent} of a matching $M$ on $[n]$ if
one of the following conditions holds:
\begin{enumerate}
\item  $\{i, i + 1\}$ is a matched pair in $M$.
\item The arc containing $i$ intersects the arc containing $i + 1$.
\item $i$ is unmatched and $i + 1$ is matched.
\end{enumerate}

Moreover, Marmor \cite{Marmor} called a matched pair $\{i,i+1\}$  in a matching {\it a short chord} and proved that the set of matchings of the set $[n]$ is Schur-positive with respect to the descent function that takes into account only short chords.

\medskip

Motivated by these works,  we define a descent parameter on set partitions, denoted ${\rm Des}$, which counts the number of elements $i$, such that $i$ and $i+1$ share the same block.

We use one of the existing definitions of a crossing number of set partitions that can be found in the literature (namely the one defined in Definition \ref{def:crossing of set partition} below) to extrapolate between Schur-positivity on the entire set of set partitions of $[n]$, denoted ${\mathbf S}et(n)$, and the set of non-crossing set partitions, denoted $\mathcal{NC}(n)$. The former set uses the Touchard-Riordan polynomials as coefficients of the linear combination into Schur functions, while the later uses the Motzkin triangle for this very purpose.

It is worth noting that we have many more results on Schur-positivity of various subsets of set partitions of $[n]$, with respect to another type of descent function. They are gathered in a paper, which is in final steps of preparation.

\medskip

This paper is organized as follows.
In Section \ref{Background}, we provide a background on Schur-positivity, crossing number of a set partition and its connection to the Touchard-Riordan polynomials and the Motzkin triangle.
In Section \ref{section on schur positivity for des sharing block}, we prove our main results.

\section{Background}\label{Background}

\subsection{Symmetric functions, quasi-symmetric functions and Schur positivity}\label{sec:prel_quasi}

In this subsection, we follow \cite{ABR,HPS}.
Let ${\bf x}=\{x_1,x_2,\dots\}$ be a countably infinite set of commuting variables and consider the algebra of formal power series over $\mathbb{Q}$.

A power series $f\in \mathbb{Q}[[\bf{x}]]$ is called {\em symmetric} if it has a bounded degree and it is invariant under permutation of variables.
Denote by ${\rm Sym}_n$ the vector space of symmetric functions, homogeneous of degree $n$.  Bases for ${\rm Sym}_n$  are indexed by {\em partitions $\lambda$ of $n$} (denoted $\lambda \vdash n$), or equivalently, by Young diagrams with $n$ boxes.

The celebrated Schur basis of ${\rm Sym}_n$ is essential for our needs.
Given a partition $\lambda \vdash n$, a {\em standard Young tableau (SYT) of shape $\lambda$} is obtained by filling the boxes of the corresponding Young diagram bijectively with the elements of $\{1,\dots,n\}$ so that rows and columns increase.  In a {\em semistandard Young tableau (SSYT) of shape $\lambda$}, we demand that the entries are positive integers such that rows weakly increase and columns strictly increase.
We write ${\rm SYT}(\lambda)$ or ${\rm SSYT}(\lambda)$ for the set of standard or semistandard Young tableaux of shape $\lambda$, respectively.

\medskip

In the following definition, we present the notion of a {\it descent set} of an SYT $T$:
\begin{definition}\label{Des of T}
The  {\em descent set of a standard Young tableau $T$} is defined to be:
$$
\Des( T) =\{i \mid \text {$i+1$ is in a lower row than $i$ in $T$}\}.
$$
\end{definition}
For example, if
$T=\raisebox{4mm}{\begin{smallytableau} 1&3&5 \\ 2&4&7 \\ 6 \end{smallytableau}}\ ,$
then $\Des (T)=\{1,3,5\}$.

\medskip

The {\it Schur function corresponding to a partition $\lambda \vdash n$} is: $s_{\lambda}=\sum\limits_{T \in {\rm SSYT}(\lambda)}\prod\limits_{i \in T}x_i.$
For example, the set ${\rm SSYT}(3,1)$ contains inter alia the following semistandard fillings of the shape
$\lambda=
{\begin{smallytableau} {}&{}&{} \\ {} \end{smallytableau}}
$\ :
$$
T_1 = {\begin{ytableau} {1}&{1}&{1} \\ {2} \end{ytableau}}\ , \ \
T_2= {\begin{ytableau} {1}&{2}&{2} \\ {2} \end{ytableau}}
\ , \ \  T_3={\begin{ytableau} {1}&{2}&{3} \\ {4} \end{ytableau}},
$$
which contribute to $s_{(3,1)}$ the monomials $m_1=x_1^3x_2,\ m_2=x_1x_2^3$ and $m_3=x_1x_2x_3x_4$ respectively.
It is well-known \cite[Chap. 7]{EC2} that the set $\{s_{\lambda}\}_{\lambda \vdash n}$ is a basis of ${\rm Sym}_n$.

For each $0\leq k\leq n$, define the {\it hook diagram} $\lambda_k=(n-k,1^k)$ (where $1^k$ stands for $k$ parts of size $1$), and note that: $$\{{\rm Des}(T) : T\in {\rm SYT}(\lambda_k)\}=\{D\subseteq [n-1]: |D|=k\}.$$

\medskip

A symmetric function is called {\em Schur-positive} if all the
coefficients in its expansion in the basis of Schur functions are nonnegative.
Determining whether a given symmetric function is
Schur-positive is a major problem in contemporary algebraic
combinatorics~\cite{Stanley_problems}.

\medskip

Quasi-symmetric functions extend symmetric functions. Here is the formal definition:
\begin{definition}\label{def quasi-symmetric}
    A {\em quasi-symmetric function} is  a formal power series $g$ of bounded degree,  satisfying that any two of its monomials $x_{i_1}^{n_1}\dots x_{i_k}^{n_k}$ (where $i_1<\dots<i_k$) and $x_{j_1}^{n_1}\dots x_{j_k}^{n_k}$ (where $j_1<\dots<j_k$) have the same coefficient in $g$.
\end{definition}
Clearly, every symmetric function is quasi-symmetric, but not conversely; for example, $\sum\limits_{i<j}{x_i^2 x_j}$  is quasi-symmetric but not symmetric.

The fundamental quasi-symmetric functions, which are defined now, serve as a basis for the vector space of homogeneous quasi-symmetric functions of degree $n$:
\begin{definition}\label{def fund}
For each subset $D \subseteq [n-1]$ define the {\em
fundamental quasi-symmetric function}
\[
\mathcal{F}_{n,D}({\bf x}) := \sum_{i_1\le i_2 \le \ldots \le i_n \atop {i_j <
i_{j+1} \text{ if } j \in D}} x_{i_1} x_{i_2} \cdots x_{i_n}.
\]
\end{definition}

Let $\BBB$ be a (multi-)set of combinatorial objects, equipped with a {\em descent function} ${\rm Des}: \BBB \to P([n-1])$, which associates to each
element $b\in \BBB$ a subset ${\rm Des}(b) \subseteq [n-1]$. Define the
quasi-symmetric function
\[
\Q_n(\BBB) := \sum\limits_{b\in \BBB} m(b,\BBB) \F_{n,\Des(b)},
\]
where $m(b,\BBB)$ is the multiplicity of the element $b \in \BBB$.

Gessel showed that (see \cite[Theorem 7.19.7]{EC2}):
\begin{thm}[Gessel]\label{gessel}
For every partition $\lambda \vdash n$,
$\Q_n({{\rm SYT}(\lambda)})=s_{\lambda}$.
\end{thm}

For another substantial example, given $A \subseteq {\mathbf S}_n$, where ${\mathbf S}_n$ is the Symmetric group on $n$ elements, define the
quasi-symmetric function $\Q_n(A) := \sum\limits_{\pi\in A} \F_{n,\Des(\pi)},$
where $\Des(\pi):=\{i:\ \pi(i)>\pi(i+1)\}$ is the ordinary descent set of $\pi$.

\subsection{The crossing number and Touchard-Riordan polynomials}
Let $n$ be a positive integer and let ${\mathbf S}et(n)$ be the set of all set partitions of $[n]$.
 Let\break $\pi=\{ B_1, B_2 , \dots , B_k\} \in {\mathbf S}et(n)$ be a set  partition of  $[n]$ into blocks $B_1,\dots,B_k$, and assume that the elements in each block $B_i$ are ordered increasingly.
The partition $\pi$ can be depicted graphically as a linear graph
made of the $n$ points
$\{1,\dots,n\}$, drawn on a (virtual) horizontal line, ordered increasingly by their labels. For every block $B_j=\{j_1,\dots,j_t\}$, we join every two points which correspond to consecutive elements, $j_k$ and $j_{k+1}$, from the same block by a semicircular arc. In the case that $j_{k+1}=j_k+1$, we join $j_k$ and $j_{k+1}$ by a straight line instead of an arc.

\begin{example}\label{exam_partition_crossing}
Let $\pi=\{\{1,3,7\},\{2,4,5,6,8\}\}$ be a partition of the set $[8]$. Then its linear graph is presented in Figure \ref{fig linear graph}.

\begin{figure}[H]
\begin{center}
\begin{tikzpicture}[scale=1]

  \foreach \x in {1,...,8} {
    \fill (\x,0) circle (2pt);
    \node[below] at (\x,0) {\x};
  }

\fill[gray] (2.32,0.67) circle (2pt);
\fill[gray] (3.5,0.77) circle (2pt);
\fill[gray] (6.59,0.71) circle (2pt);

  \def\arc#1#2#3{
    \draw (#1,0) arc[start angle=180, end angle=0, x radius={(#2-#1)/2}, y radius=#3];
  }

  \arc{1}{3}{0.7}
  \arc{3}{7}{1.2}
  \arc{2}{4}{0.9}
  \draw (4,0) -- (6,0);
  \arc{6}{8}{0.8}
\end{tikzpicture}
\end{center}
\caption{An example of a linear graph}\label{fig linear graph}
\end{figure}
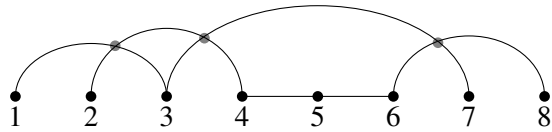

\end{example}

There are several different  definitions for the notion of crossing number of a set partition.
In this work, we choose to concentrate on the following version, see e.g. \cite{PoYa}:

\begin{definition}\label{def:crossing of set partition}
Let $\pi$ be a set partition of the set $[n]$, presented as a linear graph as above.
For $r<s$, two arcs $(i_r,j_r)$ and $(i_s,j_s)$ {\em cross} each other if $i_r<i_s<j_r<j_s$. We denote by ${\rm cr}(\pi)$ the number of such crossings and call it the {\em crossing number} of the set partition $\pi$.
\end{definition}
In Example \ref{exam_partition_crossing}, the crossings are depicted as gray points, so we have ${\rm cr}(\pi)=3$.

The crossing number can be similarly defined on matchings of $[n]$, which are a subset of ${\mathbf S}et(n)$.

\begin{definition}
A partition $\pi\in {\mathbf S}et(n)$ in which each block contains at most $2$ elements is called a {\em matching}. If $n$ is even and all the blocks are of size $2$, then $\pi$ is called a {\em perfect matching}.
A matching of the set $[n]$ can be depicted as a {\em chord diagram}, which is a circle containing $n$ points, labeled by the numbers $1,\dots, n$ ordered clockwise on the circle. Each block of size $2$ is represented by a chord connecting its two elements, see Figure \ref{crossing perfect}.
\end{definition}

\begin{figure}[H]
    \centering
\hspace{-150pt}\begin{tikzpicture}
\node[draw=none,rotate=45, minimum size=3cm,regular polygon,regular polygon sides=6] (a) {};
  \fill (a.corner 1) circle[radius=2pt] node[shift={(72+35:0.4)}] {1};
  \fill (a.corner 2) circle[radius=2pt] node[shift={(144+35:0.4)}] {6};
  \fill (a.corner 3) circle[radius=2pt] node[shift={(216+35:0.4)}] {5};
  \fill (a.corner 4) circle[radius=2pt] node[shift={(288+35:0.4)}] {4};
  \fill (a.corner 5) circle[radius=2pt] node[shift={(35:0.4)}] {3};
    \fill (a.corner 6) circle[radius=2pt] node[shift={(35:0.4)}] {2};

\draw[black]     (a.corner 1) -- (a.corner 3);
\draw[black]     (a.corner 2) -- (a.corner 5);
\draw[black]     (a.corner 4) -- (a.corner 6);

\draw (0,0) circle (1.5cm);
\end{tikzpicture}

\vspace{-120pt}\hspace{150pt}\begin{tikzpicture}
\node[draw=none,rotate=45, minimum size=3cm,regular polygon,regular polygon sides=6] (a) {};
  \fill (a.corner 1) circle[radius=2pt] node[shift={(72+35:0.4)}] {1};
  \fill (a.corner 2) circle[radius=2pt] node[shift={(144+35:0.4)}] {6};
  \fill (a.corner 3) circle[radius=2pt] node[shift={(216+35:0.4)}] {5};
  \fill (a.corner 4) circle[radius=2pt] node[shift={(288+35:0.4)}] {4};
  \fill (a.corner 5) circle[radius=2pt] node[shift={(35:0.4)}] {3};
    \fill (a.corner 6) circle[radius=2pt] node[shift={(35:0.4)}] {2};

\draw[black]     (a.corner 1) -- (a.corner 3);
\draw[black]     (a.corner 2) -- (a.corner 5);

\draw (0,0) circle (1.5cm);
\end{tikzpicture}

\caption{Left: a perfect matching, right: a non-perfect matching}
\label{crossing perfect}
\end{figure}
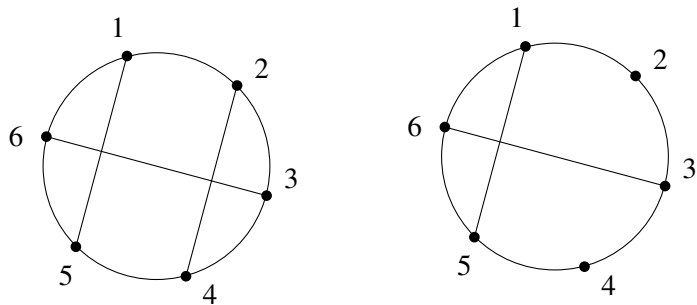

\begin{definition}
The {\em crossing number of a matching} $M$, denoted ${\rm cr}(M)$, is the crossing number of its corresponding partition as defined above, which can also be seen as the number of pairs of crossing chords in the chord diagram.
\end{definition}
The generating function of all perfect matchings on $2m$ points with respect to the crossing number is given by Touchard-Riordan polynomials, which are defined as follows:
$$\sum\limits_{M}q^{{\rm cr}(M)}=T_{2m}(q)=\frac{1}{(1-q)^m}\sum\limits_{i=0}^m(-1)^i\left( \binom{2m}{m-i}-\binom{2m}{m-i-1}\right)q^{\binom{i+1}{2}},$$
where $M$ runs through all perfect matchings on $2m$ points.
Since isolated points have no effect on the crossing number, we can easily obtain the generating function for all the matchings on $n$ points (not necessarily perfect), as follows.
Let $T_{n,j}(q)$ be the generating function for the number of matchings on $n$ points with $j$ chords with respect to the crossing number.
Then we clearly  have $p= n -2j$
unmatched points, and thus we have:
$T_{n,j}(q)=\binom{n}{2j}T_{2j}(q)$.

\subsection{Non-crossing matchings and Motzkin numbers}
\label{subsec:motzkin numbers}
Matchings $M$ satisfying ${\rm cr}(M)=0$, i.e. having no crossings, deserve special attention since they are counted by the {\it Motzkin numbers}, as originally introduced by Motzkin \cite{Motzkin} (these numbers have other various combinatorial interpretations, see \cite{DS} and Sequence A001006 in OEIS \cite{OEIS}).
The recursion defining Motzkin numbers is:
$$M_n= M_{n-1}+ \sum\limits_{i=0}^{n-2}M_i M_{n-2-i} = \frac{2n+1}{n+2}M_{n-1} + \frac{3n-3}{n+2}M_{n-2},$$
and the first Motzkin numbers are:
$M_0=1,\ M_1=1,\ 2,\ 4,\ 9,\ 21,\ 51,\ 127,\ 323.$

\begin{example}\label{Example $M_4$}
As $M_4=9$, there are $9$ ways to draw non-crossing chords between $4$ points on the circle, which are drawn in Figure \ref{fig M4=9}.

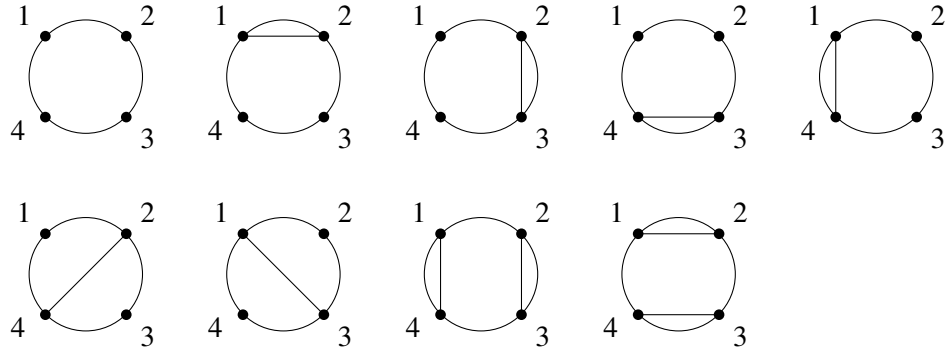
\begin{figure}[H]
\begin{center}

    \begin{tikzpicture}[
Circ/.style={draw,shape=circle,minimum size=15mm, node contents={}}
                        ]
\node (C1) [Circ];
\fill[black]    (C1.north west) circle (2pt)
                (C1.south west) circle (2pt)
                (C1.south east) circle (2pt)
                (C1.north east) circle (2pt);
\node[  shift={(100+35:0.4)},minimum size=5mm,
         ] at (C1.north west){\rm 1};
\node[  shift={(370+35:0.4)},minimum size=5mm,
         ] at (C1.north east){\rm 2};
\node[  shift={(170+35:0.4)},minimum size=5mm,
         ] at (C1.south west){\rm 4};
\node[  shift={(280+35:0.4)},minimum size=5mm,
         ] at (C1.south east){\rm 3};

\node (C2) [Circ,right=11mm of C1];
\draw[black]     (C2.north west) -- (C2.north east);
\fill[black]    (C2.north west) circle (2pt)
                (C2.south west) circle (2pt)
                (C2.south east) circle (2pt)
                (C2.north east) circle (2pt);
\node[  shift={(100+35:0.4)},minimum size=5mm,
         ] at (C2.north west){\rm 1};
\node[  shift={(370+35:0.4)},minimum size=5mm,
         ] at (C2.north east){\rm 2};
\node[  shift={(170+35:0.4)},minimum size=5mm,
         ] at (C2.south west){\rm 4};
\node[  shift={(280+35:0.4)},minimum size=5mm,
         ] at (C2.south east){\rm 3};
\node (C3) [Circ,right=11mm of C2];
\draw[black]     (C3.south east) -- (C3.north east);
\fill[black]    (C3.north west) circle (2pt)
                (C3.south west) circle (2pt)
                (C3.south east) circle (2pt)
                (C3.north east) circle (2pt);
\node[  shift={(100+35:0.4)},minimum size=5mm,
         ] at (C3.north west){\rm 1};
\node[  shift={(370+35:0.4)},minimum size=5mm,
         ] at (C3.north east){\rm 2};
\node[  shift={(170+35:0.4)},minimum size=5mm,
         ] at (C3.south west){\rm 4};
\node[  shift={(280+35:0.4)},minimum size=5mm,
         ] at (C3.south east){\rm 3};

\node (C4) [Circ,right=11mm of C3];
\draw[black]     (C4.south west) -- (C4.south east);
\fill[black]    (C4.north west) circle (2pt)
                (C4.south west) circle (2pt)
                (C4.south east) circle (2pt)
                (C4.north east) circle (2pt);
\node[  shift={(100+35:0.4)},minimum size=5mm,
         ] at (C4.north west){\rm 1};
\node[  shift={(370+35:0.4)},minimum size=5mm,
         ] at (C4.north east){\rm 2};
\node[  shift={(170+35:0.4)},minimum size=5mm,
         ] at (C4.south west){\rm 4};
\node[  shift={(280+35:0.4)},minimum size=5mm,
         ] at (C4.south east){\rm 3};
\node (C5) [Circ,right=11mm of C4];
\draw[black]     (C5.north west) -- (C5.south west);
\fill[black]    (C5.north west) circle (2pt)
                (C5.south west) circle (2pt)
                (C5.south east) circle (2pt)
                (C5.north east) circle (2pt);
\node[  shift={(100+35:0.4)},minimum size=5mm,
         ] at (C5.north west){\rm 1};
\node[  shift={(370+35:0.4)},minimum size=5mm,
         ] at (C5.north east){\rm 2};
\node[  shift={(170+35:0.4)},minimum size=5mm,
         ] at (C5.south west){\rm 4};
\node[  shift={(280+35:0.4)},minimum size=5mm,
         ] at (C5.south east){\rm 3};
\node (C6) [Circ,below=11mm of C1];
\draw[black]     (C6.south west) -- (C6.north east);
\fill[black]    (C6.north west) circle (2pt)
                (C6.south west) circle (2pt)
                (C6.south east) circle (2pt)
                (C6.north east) circle (2pt);
\node[  shift={(100+35:0.4)},minimum size=5mm,
         ] at (C6.north west){\rm 1};
\node[  shift={(370+35:0.4)},minimum size=5mm,
         ] at (C6.north east){\rm 2};
\node[  shift={(170+35:0.4)},minimum size=5mm,
         ] at (C6.south west){\rm 4};
\node[  shift={(280+35:0.4)},minimum size=5mm,
         ] at (C6.south east){\rm 3};
\node (C7) [Circ,right=11mm of C6];
\draw[black]     (C7.north west) -- (C7.south east);
\fill[black]    (C7.north west) circle (2pt)
                (C7.south west) circle (2pt)
                (C7.south east) circle (2pt)
                (C7.north east) circle (2pt);
\node[  shift={(100+35:0.4)},minimum size=5mm,
         ] at (C7.north west){\rm 1};
\node[  shift={(370+35:0.4)},minimum size=5mm,
         ] at (C7.north east){\rm 2};
\node[  shift={(170+35:0.4)},minimum size=5mm,
         ] at (C7.south west){\rm 4};
\node[  shift={(280+35:0.4)},minimum size=5mm,
         ] at (C7.south east){\rm 3};

\node (C8) [Circ,right=11mm of C7];
\draw[black]     (C8.north west) -- (C8.south west);
\draw[black]      (C8.north east) -- (C8.south east);
\fill[black]    (C8.north west) circle (2pt)
                (C8.south west) circle (2pt)
                (C8.south east) circle (2pt)
                (C8.north east) circle (2pt);
\node[  shift={(100+35:0.4)},minimum size=5mm,
         ] at (C8.north west){\rm 1};
\node[  shift={(370+35:0.4)},minimum size=5mm,
         ] at (C8.north east){\rm 2};
\node[  shift={(170+35:0.4)},minimum size=5mm,
         ] at (C8.south west){\rm 4};
\node[  shift={(280+35:0.4)},minimum size=5mm,
         ] at (C8.south east){\rm 3};
\node (C9) [Circ,right=11mm of C8];
\draw[black]     (C9.south east) -- (C9.south west);
\draw[black]      (C9.north east) -- (C9.north west);
\fill[black]    (C9.north west) circle (2pt)
                (C9.south west) circle (2pt)
                (C9.south east) circle (2pt)
                (C9.north east) circle (2pt);
\node[  shift={(100+35:0.4)},minimum size=5mm,
         ] at (C9.north west){\rm 1};
\node[  shift={(370+35:0.4)},minimum size=5mm,
         ] at (C9.north east){\rm 2};
\node[  shift={(170+35:0.4)},minimum size=5mm,
         ] at (C9.south west){\rm 4};
\node[  shift={(280+35:0.4)},minimum size=5mm,
         ] at (C9.south east){\rm 3};
\end{tikzpicture}
\end{center}
\caption{The $9$ ways to draw non-crossing chords between $4$ points on the circle}\label{fig M4=9}
\end{figure}

\end{example}

\medskip

There is a refinement of Motzkin numbers, called {\it Motzkin triangle} (see Sequence A055151 in the OEIS \cite{OEIS}), whose first 9 lines are presented in Table \ref{tab:A055151}.

\begin{table}[H]
$$\begin{NiceArray}{|w{c}{0.4cm}|c;c;c;c;c|}
\hline
\Gape[0.35cm][0.35cm]{\diagbox{\,\,n}{k\,}} & 0 & 1 & 2 & 3 & 4\\
\hline
0  &  1  & &&& \\
1  &  1  & &&& \\
2  &  1  & 1 &&& \\
3  &  1 &  3 &&&\\
4  &  1 & 6 & 2 &&\\
5  &  1 & 10 & 10 &&\\
6  &  1 & 15 & 30 & 5 & \\
7  &  1 & 21 & 70 & 35& \\
8  &  1 & 28 & 140 & 140 & 14\\
\hline
\end{NiceArray}$$

\caption{Table form of sequence A055151}
\label{tab:A055151}
\end{table}

\medskip

The element $M_{n,k}$ in the Motzkin triangle  counts the number of ways to draw $k$ non-crossing chords between $n$ points on the circle. For example, the fifth line  of the triangle ``1\ 6\ 2'' counts one diagram of $4$ points on a circle without chords, $6$ diagrams  with one chord and $2$ diagrams with 2 chords (see Example \ref{Example $M_4$}).

\begin{remark}
Note that the combinatorial explicit connection between Touchard-Riordan polynomials and Motzkin numbers is given by $T_{2n}(0)=M_{2n}$. We use this observation in Corollary \ref{specializations}(3)-(4) below.
We could not find an explicit appearance in the literature of this connection.
\end{remark}

\section{Schur-positivity of set partitions with respect to consecutive elements sharing a block}\label{section on schur positivity for des sharing block}
Let ${\mathbf S}et(n,\ell)$ be the set of all set partitions of the set $[n]$ with $\ell$ blocks (and respectively, we denote by $\mathcal{NC}(n,\ell)$ the non-crossing partitions of $[n]$ having $\ell$ blocks). For $\pi \in {\mathbf S}et(n,\ell)$, define: $${\rm Des}(\pi)=\{i \in [n-1] \mid i \mbox{ and } i+1 \mbox { are in the same block} \}.$$
Note that this set function is not sparse in the sense of Definition 1.7 of Marmor \cite{Marmor}.

\medskip

Our main result addresses the Schur-positivity of set partitions having $\ell$ blocks with respect to the parameter ${\rm  Des}$ we have just defined:

\begin{thm}
\begin{equation}\label{q touchard}
\sum\limits_{\pi \in {\mathbf S}et(n,\ell)} q^{{\rm cr} (\pi)} \mathcal{F}_{n,{\rm Des}(\pi)} =\sum\limits_{k=0}^{n-1} T_{n-1-k,n-k-\ell}(q) s_{(n-k,1^k)}
\end{equation}
\end{thm}

\begin{proof}
Let $\mathcal{D}_{n,j}$ be the set of chord diagrams having $n$ points and $j$ chords, and let ${\mathbf S}et(n,\ell,k)$ be the set of set partitions of $[n]$ with $\ell$ blocks and $k$ descents.

One can consider both sides of Equation (\ref{q touchard}) as polynomials in the variable $q$ whose coefficients are in the ring of quasi-symmetric functions. To prove the equality, we have to show that the corresponding coefficients are equal. We do that with the aid of Gessel's Theorem \ref{gessel}, by invoking a bijection
$${\rm SYT}(n-k,1^k)\times \mathcal{D}_{n-1-k,n-k -\ell} \to {\mathbf S}et(n,\ell,k),$$ which preserves the crossing number.

We start with the case $k=0$, i.e., we present a bijection from the set of pairs $$P=\left\{\left(\raisebox{-.4ex}{\begin{smallytableau} {1}&{2}&\cdots& n  \end{smallytableau}}\ , \ C\right) \mid C \in \mathcal{D}_{n-1,n-\ell}\right\}$$ to ${\mathbf S}et(n,\ell,0)$
which, as defined above, is the set of set partitions of $[n]$ with $\ell$ blocks and no descents. Let $C\in \mathcal{D}_{n-1,n-\ell}$, i.e., a chord diagram on $n-1$ points with $n-\ell$  chords. The bijection has two steps: First, for each chord $\{i,j\}$ of $C$ with $i<j$, we create the block $\{i,j+1\}$. Next, we unify every two intersecting blocks. The points unmatched by chords contribute singletons to the resulting set partition.

Since the crossing number depends solely on the arcs that connect consecutive elements in the blocks of the partition, it is easy to see that the resulting set partition, consisting of $\ell$ blocks, has the same number of crossings as in the chord diagram.

In the opposite direction, when we are given a set partition with no descents, it is straightforward to recover the chord diagram which is mapped to it.

For example, for the chord diagram in the left part of Figure \ref{example of bijection_cross}, which is an element of $\mathcal{D}_{7,3}$, we first construct the blocks $\{1,4\},\{2,6\}, \{4,8\}$, and then we send it to the set partition $$\{\{1,4,8\},\{2,6\},\{3\},\{5\}, \{7\}\}\in {\mathbf S}et (8,5,0).$$
Note that there are two crossings in both the chord diagram and the corresponding set partition.

\begin{figure}[H]
\hspace{-10pt}
\vspace{-80pt}
\begin{tikzpicture}
\node[draw=none,rotate=45, minimum size=3cm,regular polygon,regular polygon sides=7] (a) {};
  \fill (a.corner 1) circle[radius=2pt] node[shift={(72+35:0.4)}] {1};
  \fill (a.corner 2) circle[radius=2pt] node[shift={(144+35:0.4)}] {7};
  \fill (a.corner 3) circle[radius=2pt] node[shift={(216+35:0.4)}] {6};
  \fill (a.corner 4) circle[radius=2pt] node[shift={(288+35:0.4)}] {5};
  \fill (a.corner 5) circle[radius=2pt] node[shift={(35:0.4)}] {4};
  \fill (a.corner 6) circle[radius=2pt] node[shift={(35:0.4)}] {3};
  \fill (a.corner 7) circle[radius=2pt] node[shift={(35:0.4)}] {2};

\draw[black]     (a.corner 1) -- (a.corner 6);
\draw[black]     (a.corner 7) -- (a.corner 4);
\draw[black]     (a.corner 5) -- (a.corner 2);
\draw (0,0) circle (1.5cm);
\end{tikzpicture}

\hspace{150pt}\begin{tikzpicture}[scale=1]
  \foreach \x in {1,...,8} {
    \fill (\x,0) circle (2pt);
    \node[below] at (\x,0) {\x};
  }

\fill[gray] (2.72,0.69) circle (2pt);
\fill[gray] (5.37,0.66) circle (2pt);

  \def\arc#1#2#3{
    \draw (#1,0) arc[start angle=180, end angle=0, x radius={(#2-#1)/2}, y radius=#3];
  }

  \arc{1}{4}{0.7}
  \arc{4}{8}{0.7}
  \arc{2}{6}{0.9}
\end{tikzpicture}

\vspace{30pt}
\caption{Left: A chord diagram. Right: The linear graph of the resulting set partition.}
\label{example of bijection_cross}
\end{figure}
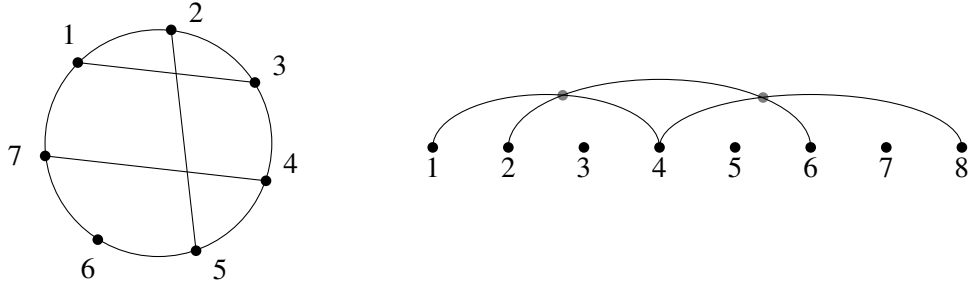

We turn now to the case $k \geq 1$. Clearly, every standard Young tableau of shape $(n-k,1^k)$ has exactly $k$ descents. Given a pair $(T,C)\in {\rm SYT}(n-k,1^k) \times \mathcal{D}_{n-1-k,n-k -\ell}$, we first define a function $\varphi:[n]\rightarrow [n-k]$ by:
$$\varphi(j)=j-|\{i\in {\rm Des}(T) \mid i<j\}|.$$
Next, we apply the procedure described above for the case $k=0$ for the chord diagram $C$ having $n-1-k$ points on a circle with $n-k-\ell$ chords, resulting in a partition of the set $[n-k]$, having $\ell$ blocks. Finally, we obtain a set partition of $[n]$ by replacing each element of $[n-k]$ by the elements in its preimage by $\varphi$, see Example \ref{example bijection k descents} below.

\medskip

Note that in the last step, while inflating the set $[n-k]$ to the set $[n]$ using $\varphi$, the crossing number of the set partition does not change, as we count only crossings of arcs, while the inflation contributes only straight segments.

\medskip

In the opposite direction, note that when we are given a set partition, one can easily recover the descent set, and the associated function $\varphi$. Then, it is straightforward to recover also the chord diagram which is mapped to it.
\end{proof}

\begin{example}\label{example bijection k descents}
Given a standard Young tableau $T$, with ${\rm Des}(T)=\{1,2,5,7,9\}\subseteq [13]$ and a chord diagram $C$ as in Figure \ref{example of bijection_cross} above, we first define
$\varphi: [13] \to [8]$, using the tableau $T$, as follows: $$\varphi(1)=\varphi(2)=\varphi(3)=1, \ \varphi(4)=2,\ \varphi(5)=\varphi(6)=3,\ \varphi(7)=\varphi(8)=4,$$
$$\varphi(9)=\varphi(10)=5,\ \varphi(11)=6, \ \varphi(12)=7 \mbox{ and } \varphi(13)=8.$$
Next, we use the chord diagram $C$ to obtain the set partition, as computed in the case $k=0$:
$$\{\{1,4,8\},\{2,6\},\{3\},\{5\}, \{7\}\}.$$
Finally, we replace each element of the set $[8]$ by the elements in its preimage by $\varphi$, to obtain the set partition:
$$\{\{\underbrace{1,2,3}_{\varphi^{-1}(1)},\underbrace{7,8}_{\varphi^{-1}(4)},\underbrace{13}_{\varphi^{-1}(8)}\},\{\underbrace{4}_{\varphi^{-1}(2)},\underbrace{11}_{\varphi^{-1}(6)}\},\{\underbrace{5,6}_{\varphi^{-1}(3)}\},\{\underbrace{9,10}_{\varphi^{-1}(5)}\},
\{\underbrace{12}_{\varphi^{-1}(7)}\}\}\in {\mathbf S}et(13,5,5).$$
\end{example}

\medskip

Plugging $q=1$ in Equation (\ref{q touchard}) induces an explicit Schur-positivity result for the set of set partitions, while substituting $q=0$ reduces us to a Schur-positivity result for non-crossing partitions, where elements from the Motzkin triangle replace the  coefficients of the Touchard-Riordan polynomials:

\begin{cor}\label{specializations}\ \\

\begin{enumerate}
\item $\mathcal{Q}_n({\mathbf S}et(n,\ell))=\sum\limits_{k=0}^{n-1} T_{n-1-k,n-\ell-k}(1) s_{(n-k,1^k)}$,
\item $\mathcal{Q}_n({\mathbf S}et(n))=\sum\limits_{k=0}^{n-1} \sum\limits_{\ell=1}^{n-k} T_{n-1-k,n-\ell-k}(1) s_{(n-k,1^k)}$,
\item $\mathcal{Q}_n({\mathcal{ NC}}(n,\ell))=\sum\limits_{k=0}^{n-1} M_{n-1-k,n-\ell-k} s_{(n-k,1^k)}$,
\item $\mathcal{Q}_n({\mathcal{NC}}(n))=\sum\limits_{k=0}^{n-1} M_{n-k-1} s_{(n-k,1^k)}$.
\end{enumerate}

\end{cor}

\medskip

\section*{Acknowledgments}
We would like to express our gratitude to Yuval Roichman for fruitful discussions.

\normalsize
\bibliographystyle{eptcs}
\bibliography{schur-positivity-gascom2026}

\end{document}